\documentclass[10pt,twocolumn,twoside]{IEEEtran}

\usepackage[utf8]{inputenc}
\usepackage{amsfonts}
\usepackage{amsmath}
\usepackage{amsthm}
\usepackage{comment}
\usepackage{xcolor}
\usepackage{algorithm, algpseudocode}
\usepackage{subcaption}
\usepackage{cite}
\PassOptionsToPackage{hyphens}{url}\usepackage{hyperref}

\usepackage{pgfplots}
  \pgfplotsset{compat=newest}
  \usetikzlibrary{plotmarks}
  \usetikzlibrary{arrows.meta}
  \usepgfplotslibrary{patchplots}
  \usepackage{grffile}

\newcommand{\R}{\mathbb{R}}
\newcommand{\bmat}[1]{\begin{bmatrix}#1\end{bmatrix}}
\newcommand\norm[1]{\left\lVert#1\right\rVert}
\newcommand{\bsmtx}{\left[ \begin{smallmatrix}} 
\newcommand{\esmtx}{\end{smallmatrix} \right]}

\newtheorem{theorem}{Theorem}
\newtheorem{definition}{Definition}

\newtheorem{remark}{Remark}

\newtheorem{lemma}{Lemma}
\newtheorem{assumption}{Assumption}
\newtheorem{property}{Property}

\title{\LARGE \bf
Stability Analysis 
using Quadratic Constraints for Systems with Neural Network Controllers 
}

\author{He Yin, Peter Seiler and  Murat Arcak \IEEEmembership{Fellow, IEEE}
\thanks{Funded in part by the Air Force Office of Scientific Research grant FA9550-18-1-0253, the Office of Naval Research grant N00014-18-1-2209, and  the U.S. National
	Science Foundation grant ECCS-1906164.}
\thanks{H. Yin and M. Arcak are with the University of California, Berkeley {\tt\small \{he\_yin, arcak\}@berkeley.edu.}}
\thanks{P. Seiler is with the University of Michigan,  Ann Arbor {\tt\small pseiler@umich.edu.}}
}

\begin{document}

\renewcommand*{\thefootnote}{\fnsymbol{footnote}}

\maketitle
\thispagestyle{empty}
\pagestyle{empty}
\setcounter{footnote}{1}

\begin{abstract}
  A method is presented to analyze the stability of feedback systems
  with neural network controllers. Two stability theorems are given to
  prove asymptotic stability and to compute an ellipsoidal inner-approximation
  to the region of attraction (ROA).  The first theorem addresses linear time-invariant systems, and
  merges Lyapunov theory with local (sector) quadratic constraints to
  bound the nonlinear activation functions in the
  neural network. The second theorem 
  allows the system
  to
  include perturbations such as unmodeled dynamics,
  slope-restricted nonlinearities, and time delay, using
  integral quadratic constraint (IQCs) to capture their
  input/output behavior. This in turn allows for off-by-one IQCs to refine the description of activation functions by capturing their slope restrictions. Both results rely
  on semidefinite programming to approximate the ROA.  The  method is
  illustrated on systems with neural networks trained to stabilize a
  nonlinear inverted pendulum as well as vehicle lateral dynamics with
  actuator uncertainty.
\end{abstract}


\section{Introduction}

The paradigm of stabilizing dynamical systems with Neural Network (NN) controllers \cite{Narendra93} has been revived following recent
development in deep NNs, e.g. policy gradient
\cite{Williams1992,Peters2008, Schulman2015, Kakade01} and behavioral cloning \cite{ALVINN1989}. 
However, feedback systems with
NN controllers suffer from lack of stability and safety certificates due to the complexity of the NN
structure. Specifically, NNs have various types of nonlinear
activation functions, potentially numerous layers, and a large number of
hidden neurons, making it difficult to apply classical
analysis methods, e.g. Lyapunov theory \cite{Khalil:2002}. Monte-Carlo simulations can be used to
investigate stability but lack formal guarantees, which
are important in safety-critical applications.


Several works propose  using quadratic constraints
(QCs) to bound the nonlinear activation
functions.  This approach is used to outer-bound the outputs of a
(static) NN given a set of inputs in \cite{2019Fazlyab} and
upper-bound the Lipschitz constant of NNs in
\cite{FazlyabLip,Pauli2020}. The work \cite{Haimin2020} uses this
 idea for finite-time reachability analysis of a
 system with a NN controller. 
The work \cite{Lavaei2018}
performs stability analysis by constructing QCs from the bounds of partial
gradients of NN controllers. Reference \cite{2018Kim} assesses global asymptotic stability
of dynamic neural network models using
QCs and Lyapunov theory.

This paper presents two main stability results for  a
feedback system with a NN controller.  
Theorem~\ref{thm:NominalLyap} provides a condition to prove 
stability and to inner-approximate the region of attraction (ROA) for a 
linear time-invariant (LTI) plant.  It uses Lyapunov theory, and local
(sector) QCs to bound the nonlinear activation
functions in the NN. Theorem~\ref{thm:RobustLyap} allows the plant to include
perturbations such as unmodeled dynamics, slope-restricted
nonlinearities, and time delay,  characterizing them with integral quadratic constraints
(IQCs) \cite{Megretski:97,Veenman:16}. This in turn allows for the use of off-by-one IQCs \cite{Lessard2015} to capture the slope restrictions of activation functions.  Both results rely on
semidefinite programming to approximate the ROA.

The specific contributions of this paper are three-fold. First, our
nominal analysis with LTI plants and NN controllers uses \textit{offset} local sector
QCs that are centered around the equilibrium inputs to the NN
activation functions and allow for analyzing stability around a
non-zero equilibrium point. 
Second, our analysis
of uncertain plants and NN controllers provides robustness
guarantees for the feedback system. The uncertain plant is modeled as
an interconnection of the nominal plant and perturbations that are
described by IQCs. The use of IQCs also allows for
plants that are not necessarily LTI. Third, the proposed framework allows for local (dynamic) off-by-one IQCs to further sharpen the description of activation functions by capturing their slope restrictions. This differs from \cite{2019Fazlyab, FazlyabLip,Pauli2020,Haimin2020,Lavaei2018}, which derive only static QCs for activation functions.

Local (static) sector IQCs have been used in the stability analysis of linear systems with actuator saturation \cite{FANG_saturation_08, Tarbouriech2014}, and unbounded nonlinearities \cite{summersCDC}. The description of these nonlinearties are refined by incorporating soft (dynamic) IQCs in the stability analysis framework for linear systems \cite{Fetzer2018_softIQC}, and polynomial systems \cite{IANNELLI2019}. Compared with these works, this work is specialized to NN-controlled systems: it exploits the specific properties of NNs and uses the Interval Bound Propagation method \cite{Gowal2018} to derive non-conservative static and dynamic local IQCs to describe NN controllers; and it also allows for the analysis of NN-controlled nonlinear systems by accommodating perturbations.

The paper is organized as follows. Section \ref{sec:nomanaly} presents
the problem formulation and the nominal stability analysis when the
plant is LTI. Section \ref{sec:robustanaly} addresses
uncertain systems using  IQCs. Section
\ref{sec:example} provides numerical examples, including a
nonlinear inverted pendulum and an uncertain vehicle model. 


\emph{Notation:}  $\mathbb{S}^n$ denotes the set of $n$-by-$n$
symmetric matrices. $\mathbb{S}_+^n$ and $\mathbb{S}_{++}^n$ denote
the sets of $n$-by-$n$ symmetric, positive semidefinite and positive
definite matrices, respectively.  $\mathbb{RL}_{\infty}$ is the set of
rational functions with real coefficients and no poles on the
unit circle.  $\mathbb{RH}_{\infty} \subset \mathbb{RL}_{\infty}$
contains functions that are analytic in the closed exterior of the
unit disk in the complex plane.  $\ell^{n_x}_2$ is the set of
sequences $x: \mathbb{N} \rightarrow \R^{n_x}$ with
$\norm{x}_2 := \sqrt{\sum_{k=0}^\infty x(k)^\top x(k)} < \infty$.  When applied to vectors, the orders $>, \leq$ are applied elementwise.  For $P \in \mathbb{S}_{++}^{n}$,
$x_* \in \R^{n}$, define the ellipsoid
\begin{align} \label{eq:epsil_def}
\mathcal{E}(P,x_*):= \{x \in \R^n : (x-x_*)^\top P (x-x_*) \leq
1\}. 
\end{align}


\section{Nominal Stability Analysis} \label{sec:nomanaly}

\subsection{Problem Formulation}

Consider the feedback system consisting of a plant $G$ and
state-feedback controller $\pi$ as shown in
Figure~\ref{fig:NominalFeedback}. As a first step, we assume the plant $G$ is a linear,
time-invariant (LTI) system defined by the following discrete-time model:
\begin{align}
\label{eq:NominalSys}
x(k+1) &= A_G\ x(k) + B_{G} \ u(k),
\end{align} 
where $x(k) \in \R^{n_G}$ is the state, $u(k) \in \R^{n_u}$ is the
 input,
$A_G\in \R^{n_G \times n_G}$ and $B_G\in \R^{n_G \times n_u}$. The
controller $\pi: \R^{n_G} \rightarrow \R^{n_u}$ is an $\ell$-layer,
feed-forward neural network (NN) defined as:
\begin{subequations}\label{eq:NNlong}
\begin{align}
w^0(k) &= x(k), \\
\label{eq:NNlong_wi}
w^i(k) &= \phi^i\left(\ W^i w^{i-1}(k) + b^i \ \right), 
           \ i = 1, \ldots , \ell, \\
u(k) &= W^{\ell+1} w^{\ell}(k) + b^{\ell+1},
\end{align}
\end{subequations}
where $w^i \in \R^{n_i}$ are the outputs (activations) from the
$i^{th}$ layer and $n_0 = n_G$.  The operations for each layer are defined by a weight
matrix $W^i \in \R^{n_i \times n_{i-1}}$, bias vector
$b^i \in \R^{n_i}$, and activation function
$\phi^{i}: \R^{n_i} \rightarrow \R^{n_i}$. The activation function
$\phi^i$ is applied element-wise, i.e.
\begin{align}
\phi^i(v) := \bmat{\varphi(v_1), \cdots, \varphi(v_{n_i})}^\top,
\end{align}
where $\varphi: \R \rightarrow \R$ is the (scalar) activation function
selected for the NN. Common choices for the scalar
activation function include $\varphi(\nu):=\tanh(\nu)$,
sigmoid $\varphi(\nu):=\frac{1}{1+e^{-\nu}}$, ReLU
$\varphi(\nu):=\max(0,\nu)$, and leaky ReLU
$\varphi(\nu):=\max(a\nu,\nu)$ with $a\in (0,1)$. We assume the  activation $\varphi$ is identical in all layers; this can be relaxed with minor changes to the notation.
\vspace{-0.2in}
\begin{figure}[h]
  \centering
  \includegraphics[width=0.4\textwidth]{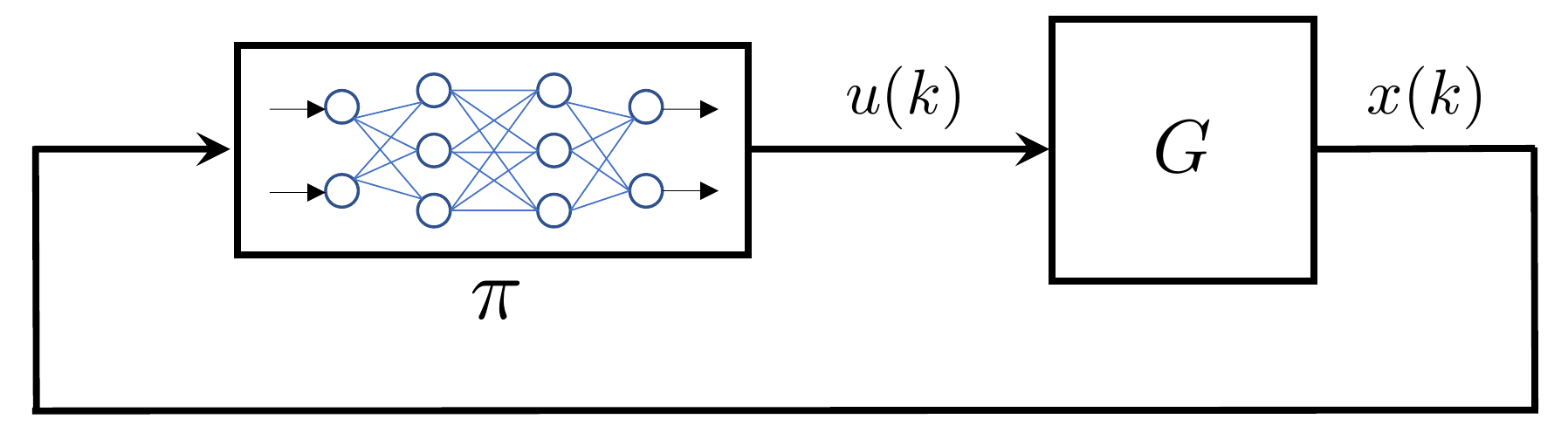}
  \caption{Feedback system with plant $G$ and NN $\pi$}
  \label{fig:NominalFeedback}    
\end{figure}
\vspace{-0.1in}

The state vector $x_*$ is an equilibrium point of the
feedback system with input $u_*$ if the following conditions hold:
\begin{subequations}\label{eq:EqPtxu}
\begin{align}
    x_* & = A_G \ x_* + B_G \ u_*, \\
    u_* & = \pi(x_*).
\end{align}
\end{subequations}
Let $\chi(k;x_0)$ denote the solution to the feedback system at time $k$
from initial condition $x(0)=x_0$. Our goal is to analyze asymptotic stability of the equilibrium point and to 
find the largest estimate of the region of attraction, defined below,  using an
ellipsoidal inner approximation.

\begin{definition}
  \label{def:ROA}
  The region of attraction (ROA) of the feedback system with plant $G$ and NN $\pi$ is defined
  as
  \begin{align}
    \mathcal{R} := \{x_0 \in \R^{n_G}: \lim_{k \rightarrow \infty} \chi(k;x_0) = x_*\}.
  \end{align}
\end{definition}


\subsection{NN Representation: Isolation of Nonlinearities}

It is useful to isolate the nonlinear activation functions from the
linear operations of the NN as done in \cite{2019Fazlyab,2018Kim}.
Define $v^i$ as the input to the activation function $\phi^i$:
\begin{align}
\label{eq:PhiInput}
  v^i(k) := W^i w^{i-1}(k) + b^i,
   \ i = 1, \ldots , \ell.
\end{align}
The nonlinear operation of the $i^{th}$ layer
(\ref{eq:NNlong_wi}) is thus expressed as
$w^i(k) = \phi^i( v^i(k) )$. Gather the inputs and
outputs of all activation functions:
\begin{align}
  v_\phi := \bmat{v^1 \\ \vdots \\ v^{\ell}} \in \R^{n_\phi} 
\mbox{ and }
  w_\phi := \bmat{w^1 \\ \vdots \\ w^{\ell}} \in \R^{n_\phi},
\end{align}
where $n_\phi:=n_1+\cdots + n_\ell$, and define the combined
nonlinearity $\phi: \R^{n_\phi} \rightarrow \R^{n_\phi}$ by stacking
the activation functions:  
\begin{align}
    \phi(v_\phi) := \bmat{\phi^1(v^1) \\ \vdots \\ \phi^\ell(v^\ell)}. 
\label{eq:phi_def}
\end{align}
Thus
$w_\phi(k) = \phi( v_\phi(k) )$, where the scalar activation function
$\varphi$ is applied element-wise to each entry of $v_\phi$.
Finally, the NN control policy $\pi$ defined in \eqref{eq:NNlong} can be rewritten as:
\begin{subequations}\label{eq:NNlft}
\begin{align}
  \bmat{u(k)\\ v_\phi(k)} & = N  \bmat{x(k) \\ w_\phi(k) \\ 1} \\
  w_\phi(k) & = \phi( v_\phi(k) ).
\end{align}  
\end{subequations}
The matrix $N$ depends on the weights and biases as follows, where the vertical and horizontal bars  partition $N$ 
compatibly with the inputs $(x,w_\phi,1)$ and outputs $(u,v_\phi)$:
\begin{subequations}\label{eq:NNmatrix}
\begin{align}
N & := \left[ \begin{array}{c|cccc|c} 
  0 & 0 &  0 & \cdots & W^{\ell+1} &  b^{\ell+1} \\ \hline
  W^1 & 0   & \cdots & 0 & 0 & b^1  \\ 
  0   & W^2 & \cdots & 0 & 0 & b^2 \\
  \vdots & \vdots & \ddots & \vdots & \vdots & \vdots \\
  0   & 0   & \cdots & W^\ell & 0 & b^\ell 
 \end{array}\right] \\
 & := \bmat{ N_{ux} & N_{uw} & N_{ub} \\ N_{vx} & N_{vw} & N_{vb} }.
\end{align}
\end{subequations}
This decomposition of the NN, depicted in
Figure~\ref{fig:NNLFT},
 isolates the activation functions
in preparation for the stability analysis.

\begin{figure}[h]
  \centering
  \includegraphics[width=0.3\textwidth]{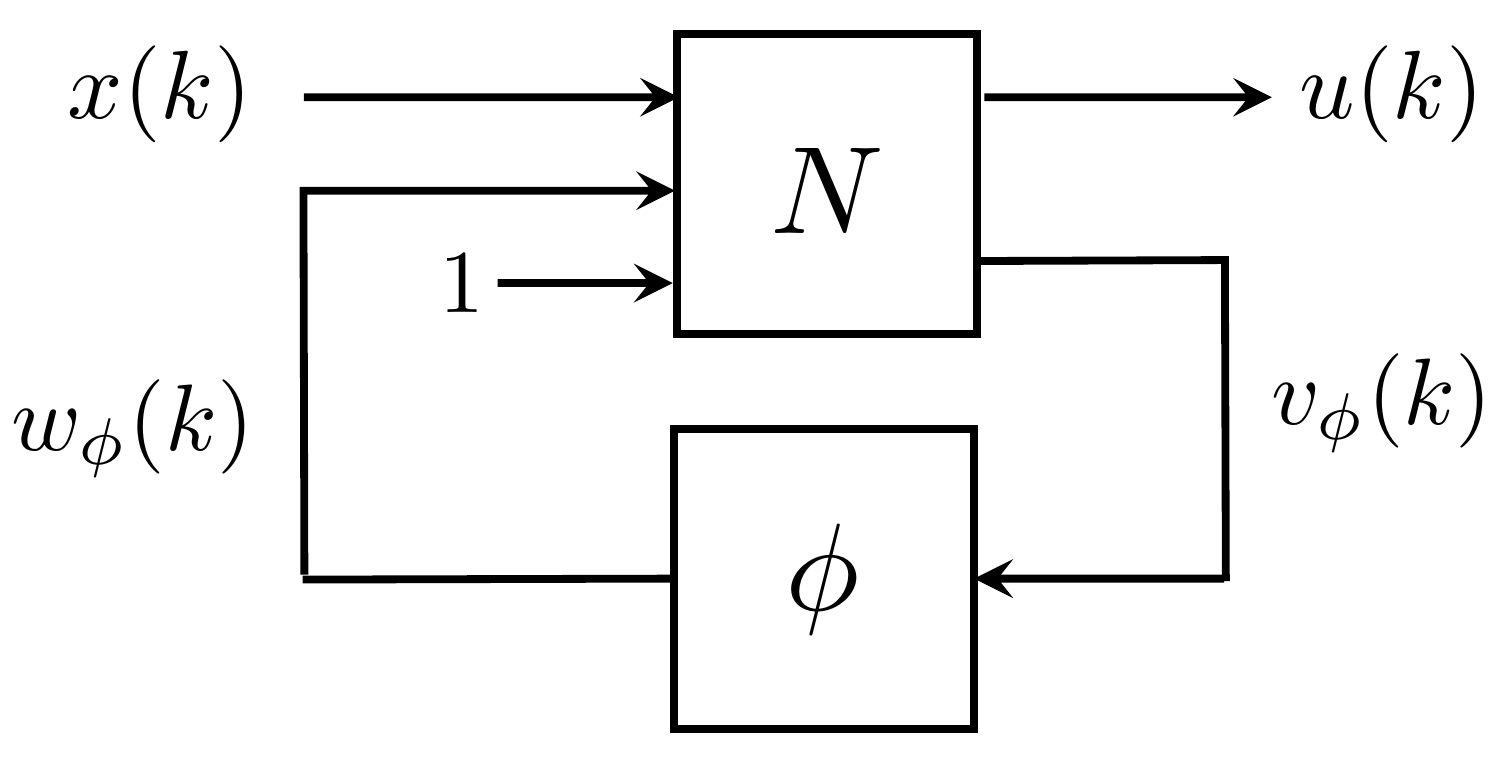}
  \caption{NN representation to isolate the nonlinearities $\phi$.}
  \label{fig:NNLFT}    
\end{figure}

Suppose $(x_*,u_*)$  satisfies \eqref{eq:EqPtxu}.  Then 
$x_*$ can be propagated through the NN to obtain equilibrium values
$v_*^i$, $w_*^i$ for the inputs/outputs of each activation function
($i=1,\ldots, \ell$), yielding 
$(v_\phi,w_\phi)=(v_{*},w_{*})$.  Thus $(x_*, u_*, v_{*},w_{*})$ is an
equilibrium point of \eqref{eq:NominalSys} and \eqref{eq:NNlong} if:
\begin{subequations}\label{eq:EqPtxuvw}
\begingroup
\allowdisplaybreaks
  \begin{align} 
    x_* &= A_G \ x_* + B_G \ u_*, \\
    \bmat{u_*\\v_{*}} &= N \bmat{x_* \\ w_{*} \\ 1}, \label{eq:u*v*_def}\\
    w_{*} &= \phi(v_{*}).\label{eq:w*_def}
  \end{align}
  \endgroup
\end{subequations}

\subsection{Quadratic Constraints: Scalar Activation Functions}

The stability analysis relies on quadratic constraints (QCs) to bound 
the activation function.  
A typical
constraint is the sector bound as defined next.
\begin{definition}
  \label{def:GlobalSector}
  Let $\alpha \le \beta$ be given.  The function
  $\varphi: \R \rightarrow \R$ lies in the (global) sector
  $[\alpha,\beta]$ if:
  \begin{align}
    ( \varphi(\nu) - \alpha \nu ) \cdot
       (\beta \nu - \varphi(\nu)) \ge 0
    \,\,\, \forall \nu \in \R.
  \end{align}
\end{definition}
The interpretation of the sector $[\alpha,\beta]$ is  that
$\varphi$ lies between lines passing through the origin with slope
$\alpha$ and $\beta$.  Many activation functions are bounded in the
sector $[0,1]$, e.g. $\tanh$ and ReLU. 
Figure~\ref{fig:tanh} illustrates $\varphi(\nu) = \tanh(\nu)$ (blue solid)
and the global sector defined by $[0,1]$ (red solid lines).
\vspace{-0.1in}
\begin{figure}[h]
  \centering
  \includegraphics[width=0.42\textwidth]{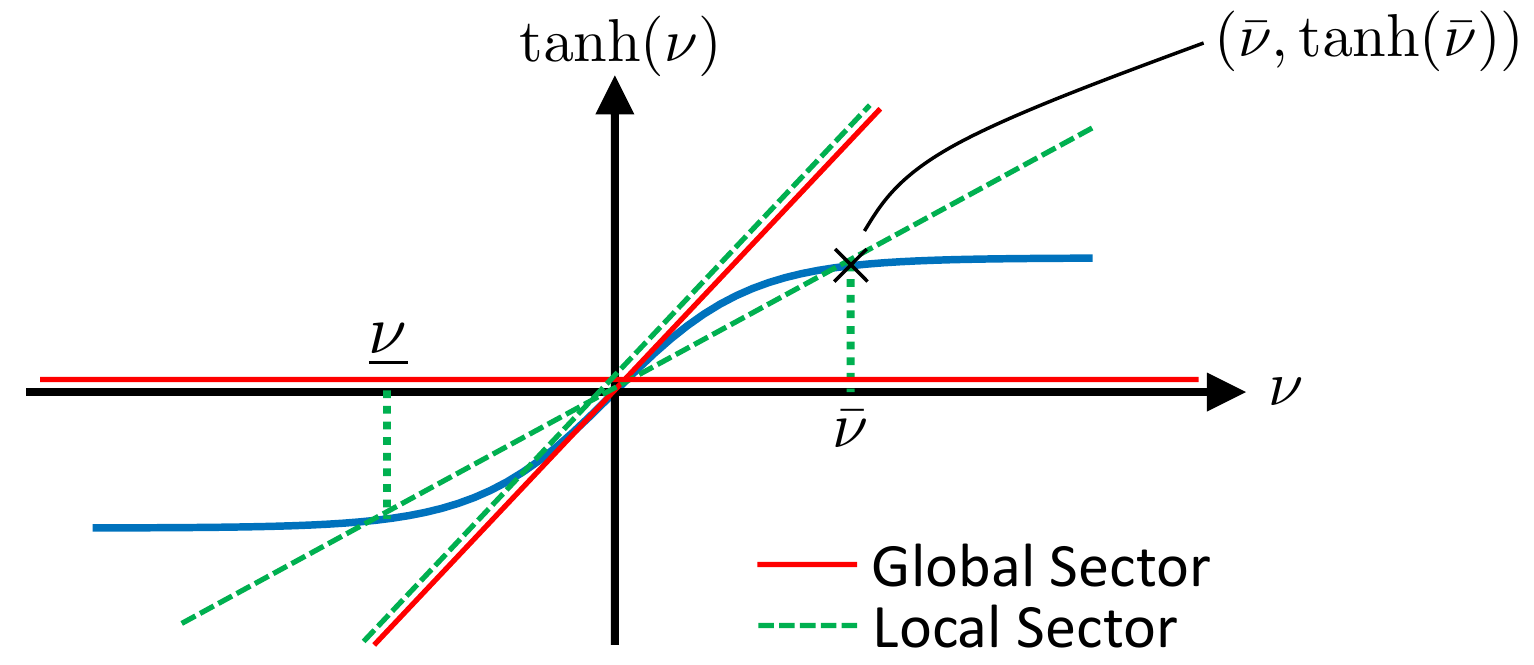}
  \caption{Sector constraints on $\tanh$}
  \label{fig:tanh}    
\end{figure}

The global sector constraint is often too coarse for
stability analysis, and a local sector constraint  provides
tighter bounds.

\begin{definition}
  \label{def:LocalSector}
  Let $\alpha$, $\beta$, $\underline \nu$, $\bar \nu \in\R$ 
  with $\alpha \le \beta$ and $\underline \nu \le 0 \le \bar \nu$.
  The function $\varphi: \R \rightarrow \R$ satisfies the local sector
  $[\alpha,\beta]$ if
  \begin{align}
    ( \varphi(\nu) - \alpha \, \nu ) \cdot
     (\beta \, \nu - \varphi(\nu)) \ge 0
    \,\,\, \forall \nu \in [\underline{\nu},\bar{\nu} ].
  \end{align}
\end{definition}

As an example,  $\varphi(\nu):=\tanh(\nu)$ restricted to the
interval $[-\bar{\nu},\bar{\nu}]$  satisfies the local
sector bound $[\alpha,\beta]$ with
$\alpha:=\tanh(\bar\nu)/\bar{\nu}>0$ and
$\beta:=1$. As shown in Figure~\ref{fig:tanh}  (green
dashed lines), the local sector provides a tighter bound than the global sector.  These bounds are
valid for a symmetric interval around the origin with
$\underline\nu = -\bar\nu$; non-symmetric intervals ($\underline\nu \ne -\bar \nu$) can be handled similarly.

The local and global sector constraints above were defined to be
centered at the point $(\nu,\varphi(\nu))=(0,0)$. The stability
analysis will require offset sectors centered around an arbitrary
point $(\nu_*,\varphi(\nu_*))$ on the function.  For example,  $\varphi(\nu)=\tanh(\nu)$ satisfies the global
sector bound (red solid) around the point $(\nu_*,\tanh(\nu_*))$ with
$[\alpha, \beta]=[0,1]$, as shown in
Figure~\ref{fig:tanhshifted}. It satisfies a tighter
local sector bound (green dashed) when the input is restricted to
$\nu \in[\underline{\nu},\bar{\nu}]$.  An explicit expression for this
local sector is $\beta=1$ and
\begin{align*}
\alpha := \min\left( \frac{\tanh(\bar{\nu})-\tanh(\nu_*)}{\bar{\nu} -
  \nu_*}, \frac{\tanh(\nu_*)-\tanh(\underline{\nu})}{\nu_* -
  \underline{\nu}} \right).
\end{align*}
The local sector upper bound $\beta$ can be tightened further.  This
leads to the following definition of an offset local sector.

\begin{figure}[h]
  \centering
  \includegraphics[width=0.42\textwidth]{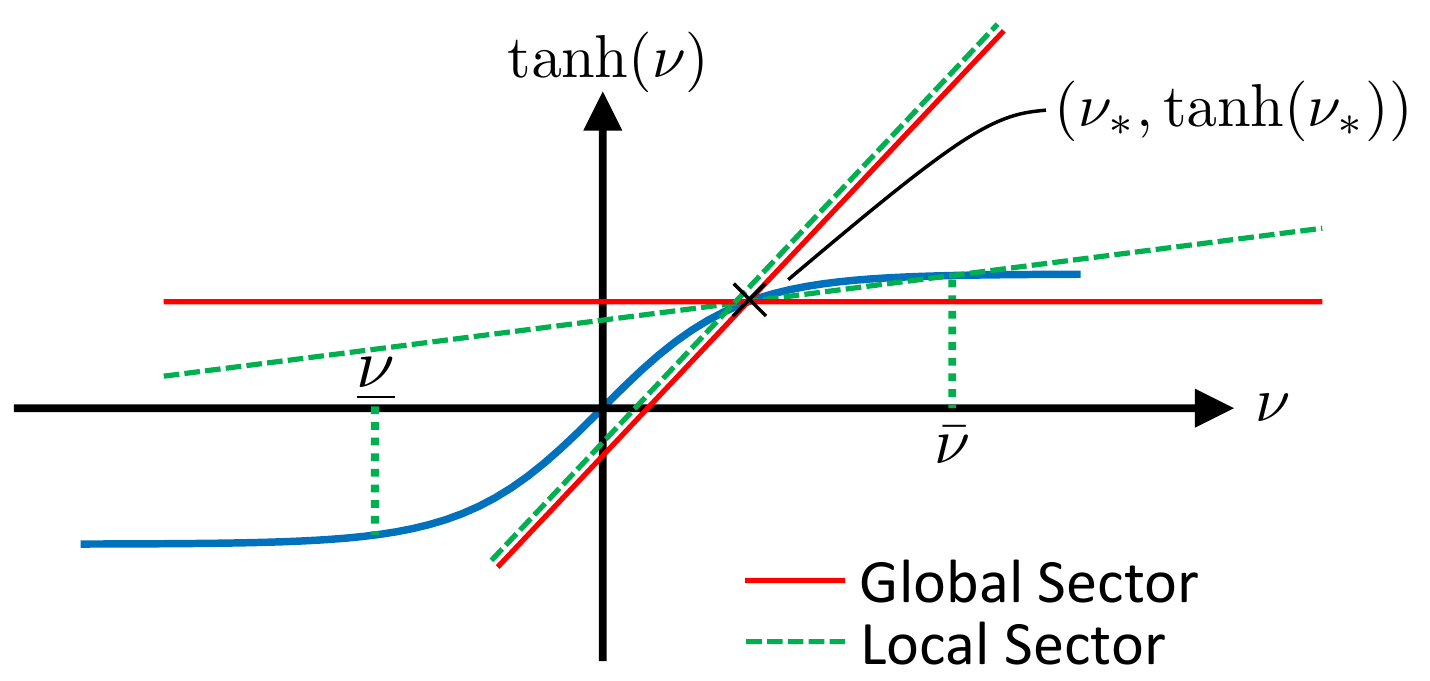}
  \caption{Offset local sector constraint on $\tanh$}
  \label{fig:tanhshifted}    
\end{figure}

\begin{definition}
  \label{def:LocalSectorOffset}
  Let $\alpha$, $\beta$, $\underline \nu$, $\bar \nu$, $\nu_*\in\R$ be
  given with $\alpha \le \beta$ and
  $\underline \nu \le \nu_* \le \bar \nu$.  The function
  $\varphi: \R \rightarrow \R$ satisfies the offset local sector
  $[\alpha,\beta]$ around the point $(\nu_*, \varphi(\nu_*))$ if:
  \begin{align}
    ( \Delta \varphi(\nu) - \alpha \, \Delta \nu ) \cdot
     (\beta \, \Delta \nu - \Delta \varphi(\nu)) \ge 0
    \,\,\, \forall \nu \in [\underline{\nu},\bar{\nu} ]
  \end{align}
  where $\Delta \varphi(\nu): = \varphi(\nu) - \varphi(\nu_*)$
  and $\Delta \nu : = \nu - \nu_*$.
\end{definition}

\vspace{-0.3in}
\subsection{Quadratic Constraints: Combined Activation Functions}

Offset local sector constraints can also be defined for the combined
nonlinearity $\phi$, given by \eqref{eq:phi_def}.  Let
$\underline{v}, \bar{v}, v_* \in \R^{n_\phi}$ be given with
$\underline{v} \le v_* \le \bar{v}$.  Assume that the activation input
$v_\phi \in \R^{n_\phi}$ lies, element-wise, in the interval
$[\underline{v},\bar{v}]$ and the $i^{th}$ input/output pair is
$w_{\phi,i} = \varphi( v_{\phi,i} )$.  Further assume the scalar activation function
satisfies the local sector $[\alpha_i,\beta_i]$ around the point
$v_{*,i}$ with the input restricted to
$v_{\phi,i} \in [\underline{v}_i,\bar{v}_i]$ for $i=1,\ldots,n_\phi$.  The
local sector bounds can be computed for $\varphi$ on the given
interval either analytically (as above for $\tanh$) or
numerically. These local sectors can  be stacked into vectors
$\alpha_\phi, \beta_\phi \in \R^{n_\phi}$ that provide 
QCs satisfied by the
combined nonlinearity $\phi$.
 
\begin{lemma} \label{lemma:sectorQC} Let $\alpha_\phi$, $\beta_\phi$,
  $\underline v$, $\bar v$, $v_* \in\R^{n_\phi}$ be given with
  $\alpha_\phi \le \beta_\phi$, $\underline v \le v_* \le \bar v$, and
  $w_*:=\phi(v_*)$.  Assume $\phi$ satisfies the offset local sector
  $[\alpha_\phi,\beta_\phi]$ around the point $(v_*, w_*)$ element-wise for
  all $v_\phi \in [\underline v,\bar v]$. If $\lambda \in \R^{n_\phi}$
  with $\lambda \ge 0$ then:
  \begin{align}
    & \bmat{v_\phi-v_* \\ w_\phi-w_*}^\top 
    \Psi_\phi^\top M_\phi(\lambda) \Psi_\phi
    \bmat{v_\phi-v_* \\ w_\phi-w_*} \ge 0 \nonumber\\
    & \hspace{0.4in} \forall v_\phi \in [\underline v,\bar v], 
    \, w_\phi = \phi( v_\phi ), \nonumber \\
    \text{where} \ \ \Psi_{\phi} & := \bmat{ \text{diag}(\beta_\phi) & -I_{n_\phi}  \\
     -\text{diag}(\alpha_\phi)  & I_{n_{\phi}}}  \label{eq:Psiphi0} \\
     M_\phi(\lambda) & := \bmat{0_{n_\phi} & \text{diag}(\lambda) \\ 
             \text{diag}(\lambda) & 0_{n_\phi}}. \label{eq:Mphi0}
   \end{align}
\end{lemma}
\begin{proof}
  For any $v_\phi \in \R^{n_\phi}$ and $w_\phi = \phi( v_\phi )$:
  \begingroup
\allowdisplaybreaks
  \begin{align*}
    & \bmat{v_\phi-v_* \\ w_\phi-w_*}^\top 
    \Psi_\phi^\top M_\phi(\lambda) \Psi_\phi
    \bmat{v_\phi-v_* \\ w_\phi-w_*} \\
    & \,\,\,\,
      = \sum_{i=1}^{n_\phi} \lambda_i 
      ( \Delta w_i - \alpha_i \, \Delta v_i ) \cdot
      (\beta_i \, \Delta v_i - \Delta w_i) 
  \end{align*}
  \endgroup
  where $\Delta w_i: = \varphi(v_{\phi,i}) - \varphi(v_{*,i})$ and
  $\Delta v_i : = v_{\phi,i} - v_{*,i}$.  If $v_\phi \in [\underline v,\bar v]$
  then each term in the sum is non-negative
  by the offset local sector constraints and $\lambda \ge 0$.
\end{proof}

In order to
apply the local sector and slope bounds in the stability analysis, we must first compute the bounds $\underline{v}, \overline{v} \in \R^{n_\phi}$ on the activation input $v_\phi$. The process to compute the bounds is briefly discussed here with more details provided in \cite{Gowal2018}.
Let
$v_*^1$ be the equilibrium value at the first NN layer.  Select
$\underline v^1$, $\bar v^1\in\R^{n_1}$ with
$\underline v^1 \le v_*^1 \le \bar v^1$.  
The assumed bounds on $v^1$ can be used
to compute an interval $[\underline{w}^1,\bar{w}^1]$ for the
output $w^1=\phi^1(v^1)$\footnote{For example, if
  $\varphi(\nu) = \tanh(\nu)$ then the input bound
  $\nu \in [-\bar\nu,\bar \nu]$ implies the output bound
  $\varphi(\nu) \in [-\tanh(\bar\nu),\tanh(\bar\nu)]$.} 
which can then be used to compute bounds
$[\underline{v}^2,\bar{v}^2]$ on the input $v^2$ to the next
activation function.\footnote{The next activation input is
  $v^2 := W^2 w^1 + b^2$. The largest value of the $i^{th}$ entry of
  this vector is obtained by solving the following optimization:
  \begin{align}
    \bar{v}^2_i := 
     \max_{\underline{w}^1 \le w^1 \le \bar{w}^1} y^\top w^1 + b_i^2
  \end{align}
  where $y^\top$ is the $i^{th}$ row of $W^2$. Define
  $c:=\frac{1}{2} (\bar{w}^1 + \underline{w}^1)$ and
  $r:=\frac{1}{2} (\bar{w}^1 - \underline{w}^1)$. The optimization can
  be rewritten as:
  \begin{align}
    \bar{v}^2_i := \left( y^\top c + b_i^2\right) + 
     \max_{-r \le \delta \le r} y^\top \delta
  \end{align}
  This has the explicit solution
  $\bar{v}^2_i = y^\top c + b_i^2 + \sum_{j=1}^{n_1} | y_j r_j
  |$. Similarly, the minimal value is
  $\underline{v}^2_i = y^\top c + b_i^2 - \sum_{j=1}^{n_1} | y_j r_j
  |$.}  The intervals computed for $w^1$ and $v^2$ will contain their
equilibrium value $w_*^1$ and $v_*^2$.  This process can be propagated
through all layers of the NN to obtain the bounds $\underline{v}, \overline{v} \in \R^{n_\phi}$ for the activation function input $v_\phi$. The  remainder  of  the  paper  will  assume  the  local  sector  bounds  have  been  computed  as briefly  summarized  in the following property.

\begin{property}
  \label{prop:NNsector}  
  Let $v_* \in \R^{n_\phi}$ be an equilibrium value of the activation
  input and $v_*^1 \in \R^{n_1}$ be the corresponding value at the
  first layer.  Let $\underline{v}^1$, $\bar{v}^1\in \R^{n_1}$ with
  $v^1_* \in [\underline{v}^1, \bar{v}^1]$ and their corresponding activation input bounds $\underline{v}, \overline{v}$ be given. There exist
  $\alpha_\phi$, $\beta_\phi\in \R^{n_\phi}$ such that $\phi$
  satisfies the offset local sector around the point
  $(v_{*},\phi(v_*))$ for all 
  $v_\phi \in [\underline{v}, \bar{v}]$.
\end{property}

\subsection{Lyapunov Condition}

This section uses a Lyapunov function and the offset local sector to compute an inner approximation for
the ROA of the feedback system of $G$ and $\pi$.  To
simplify notation, the interval bound on $v^1$ is assumed to be
symmetrical about $v^1_*$, i.e. $\underline{v}^1 = 2v^1_*-\bar{v}^1$ so
that $\bar{v}^1 - v_*^1 = v_*^1-\underline{v}^1$. This can be relaxed to
handle non-symmetrical intervals with minor notational changes.

\begin{theorem}
  \label{thm:NominalLyap}
  Consider the feedback system of plant $G$ in \eqref{eq:NominalSys}
  and NN $\pi$ in \eqref{eq:NNlong} with equilibrium point
  $(x_*, u_*, v_{*}, w_{*})$ satisfying \eqref{eq:EqPtxuvw}. Let
  $\bar{v}^1 \in \R^{n_1}$, $\underline{v}^1:=2v_*^1-\bar{v}^1$, and
  $\alpha_\phi, \beta_\phi \in \R^{n_\phi}$ be given vectors
  satisfying Property~\ref{prop:NNsector} for the NN.  Denote the
  $i^{th}$ row of the first weight $W^1$ by $W^1_i$ and define
  matrices
  \begin{align*}
    R_V :=\bmat{I_{n_G} & 0_{n_G \times n_\phi} \\ N_{ux} & N_{uw}}, 
    \, \mbox{ and } \,
    R_\phi := \bmat{N_{vx} & N_{vw} \\ 0_{n_\phi \times n_G} & I_{n_\phi}}.
  \end{align*}
  If there exists a matrix $P \in \mathbb{S}_{++}^{n_G}$, and vector
  $\lambda \in \R^{n_\phi}$ with $\lambda \ge 0$ such that
  \begingroup
\allowdisplaybreaks
  \begin{align} 
    \nonumber
    & R_{V}^\top \bmat{A_G^\top P A_G - P & A_G^\top P B_G 
     \\ B_G^\top P A_G & B_G^\top P B_G} R_{V}  \\
    & \hspace{0.4in}
     +  R_{\phi}^\top \Psi_\phi^\top M_\phi(\lambda) \Psi_\phi R_{\phi} < 0, 
    \label{eq:diss_nominal} \\
    &\bmat{(\bar{v}_i^{1}-v^1_{*,i})^2 & W^1_i \\ W^{1\top}_i & P}  
    \ge 0, \,\,\,i = 1,\cdots, n_1, 
    \label{eq:setcontain}
  \end{align}
  \endgroup
  then: (i) the feedback system consisting of $G$ and $\pi$ is locally
  stable around $x_*$, and (ii) the set $\mathcal{E}(P,x_*)$, defined by \eqref{eq:epsil_def}, is an 
  inner-approximation to the ROA.
\end{theorem}
\begin{proof}
  By Schur complements, \eqref{eq:setcontain} is equivalent to:
  \begin{align}
    W_i^{1} P^{-1} W_i^{1 \top} \leq (\bar{v}_i^1 - v_{*,i}^1)^2, 
      \ i = 1,\cdots,n_1.
  \end{align}  
  It follows from Lemma 1 in \cite{Hindi:98} that:
  \begin{align*}
    \mathcal{E}(P,x_*) \subseteq   \{x \in \R^{n_G}:
     \underline{v}^1 - v_*^1 \leq W^1(x - x_*) \leq \bar{v}^1 - v_*^1\}.
  \end{align*}
  Finally, use $v^1-v_*^1=W^1 (x -x_*)$  to rewrite this as:
  \begin{align*}
    \mathcal{E}(P,x_*) \subseteq \{x: \underline{v}^1 
        \leq v^1 \leq \bar{v}^1\}. 
  \end{align*}
  To summarize, feasibility of \eqref{eq:setcontain} verifies that if
  $x(k) \in \mathcal{E}(P,x_*)$ then
  $v^1(k) \in [\underline{v}^1, \bar{v}^1]$ and hence the offset local
  sector conditions are valid.

  Next, since the LMI in \eqref{eq:diss_nominal} is strict, there exists $\epsilon > 0$ such that left / right multiplication of the LMI by $\bmat{(x(k)-x_*)^\top & (w_\phi(k)-w_{*})^\top}$ and its
  transpose yields
  \begin{align*}
    &\Big[ \star \Big]^\top \bmat{A_G^\top P A_G - P & 
        A_G^\top P B_G \\ B_G^\top P A_G & B_G^\top P B_G}
        \bmat{x(k)-x_* \\ u(k)-u_*}  \\
    & + \Big[ \star \Big]^\top \Psi_\phi^\top 
       M_\phi(\lambda) \Psi_\phi \bmat{v_\phi(k)-v_{*}\\w_\phi(k)-w_{*}} \leq -\epsilon \| x(k) - x_*\|^2.
  \end{align*}
  \noindent where the entries denoted by $\star$ can be inferred from symmetry.
  Define the Lyapunov function $V(x):=(x-x_*)^\top P (x-x_*)$ and
  use \eqref{eq:NominalSys} and \eqref{eq:EqPtxuvw} to show:
  \begin{align}
    \label{eq:NominalLyapCond}
    & V(x(k+1)) - V(x(k)) + \Big[ \star \Big]^\top \Psi_\phi^\top
       M_\phi(\lambda) \Psi_\phi \bmat{v_\phi(k)-v_{*}\\w_\phi(k)-w_{*}}
    \nonumber \\
   &~~~~~~~~  \leq -\epsilon \| x(k) - x_*\|^2.
  \end{align}
  Assume $x(k) \in \mathcal{E}(P,x_*)$ for some $k \ge 0$, i.e., $V(x(k)) \leq 1$. As noted above, $x(k) \in \mathcal{E}(P,x_*)$ implies the offset local sector $[\alpha_\phi, \beta_\phi]$ around
  $v_{*}$. Then, by Lemma~\ref{lemma:sectorQC}, 
  the final term on the left side of \eqref{eq:NominalLyapCond} is $\ge 0$, and thus from \eqref{eq:NominalLyapCond} we have $V(x(k+1)) \leq 1$, i.e., $x(k+1) \in \mathcal{E}(P, x_*)$. By induction, we have that $\mathcal{E}(P, x_*)$ is forward invariant, i.e., $x(0) \in \mathcal{E}(P, x_*) \implies x(k) \in \mathcal{E}(P, x_*) \ \forall k \ge 0$. As a result, if $x(0) \in \mathcal{E}(P,x_*)$, then the final term on the left side of \eqref{eq:NominalLyapCond} is $\ge 0$ for all $k \ge 0$, and  $V(x(k+1)) - V(x(k)) \leq -\epsilon \|x(k)-x_*\|^2$ for all $k \ge 0$. It follows from a Lyapunov argument, e.g. Theorem 4.1 in
  \cite{Khalil:2002}, that $x_*$ is an asymptotically stable
  equilibrium point and $\mathcal{E}(P,x_*)$ is an inner approximation
  of the ROA.
\end{proof}
\vspace{-0.2in}

\begin{remark}
Note that $\overline{v}^1$ should be chosen with care as it affects the size of ROA inner-approximations directly: decreasing $(\overline{v}^1 - v_*^1)$ gives rise to sharper local sector bounds, which is beneficial on ROA estimation, but also restricts the region where ROA inner-approximations lie in; increasing $(\overline{v}^1 - v_*^1)$ leads to a larger region that contains ROA inner-approximations, but also provides looser local sector bounds. A possible way of choosing $\overline{v}^1$ is to parameterize $(\overline{v}^1 - v_*^1)$ as $\overline{v}^1 - v_*^1 = \delta_v \times 1_{n_1 \times 1}$ with $\delta_v \in \R_{++}$, grid the interval $[0, \overline{\delta}_v]$\footnote{$\overline{\delta}_v$ is the largest value such that \eqref{eq:diss_nominal} and \eqref{eq:setcontain} stay feasible.} where $\delta_v$ lies in, inner-approximate the ROA on the grid, and choose $\delta_v$ that leads to the largest inner-approximation.
\end{remark}

\begin{remark}
    In the paper, the NN controller is assumed to be state-feedback. For the output-feedback case, i.e., $u = \pi(C x)$, where $C \in \R^{n_y\times n_G}$, the stability analysis can be performed similarly, using a new $N_{vx}$ defined as $N_{vx}:= \left[\begin{smallmatrix}W^1 C \\ 0_{(n_2+...+n_\ell)\times n_G}\end{smallmatrix}\right]$.
\end{remark}




\section{Robust Stability Analysis}\label{sec:robustanaly}

\subsection{Problem Formulation}
Consider the uncertain feedback system in
Figure~\ref{fig:UncFeedback}, consisting of an uncertain plant
$F_u(G,\Delta)$ and a NN controller $\pi$ as defined by
\eqref{eq:NNlong}. The uncertain plant $F_u(G,\Delta)$ is an
interconnection of a nominal plant $G$ and a perturbation
$\Delta$. The nominal plant $G$ is defined by the following equations:
\begin{subequations} \label{eq:UncertainSys}
\begingroup
\allowdisplaybreaks
\begin{align}
x(k+1) &= A_G\ x(k) + B_{G1} \ q (k) + B_{G2} \ u(k), \\
p(k) &= C_G \ x(k) + D_{G1} \ q (k) + D_{G2} \ u(k),
\end{align}
\endgroup
\end{subequations}
where $x(k) \in \R^{n_{G}}$ is the state, $u(k) \in \R^{n_u}$ is the
control input, $p(k) \in \R^{n_p}$ and $q(k) \in \R^{n_q}$ are
the input and output of $\Delta$, $A_G \in \R^{n_G \times n_G}$,
$B_{G1} \in \R^{n_G \times n_q}$, $B_{G2} \in \R^{n_G \times n_{u}}$,
$C_G \in \R^{n_p \times n_G}$, $D_{G1} \in \R^{n_{p} \times n_q}$, and
$D_{G2} \in \R^{n_p \times n_{u}}$. The perturbation is a bounded,
causal operator
$\Delta :\ell^{n_{p}}_{2e} \rightarrow \ell^{n_q}_{2e}$. The nominal
plant $G$ and perturbation $\Delta$ form the interconnection
$F_u(G,\Delta)$ through the constraint
\begin{align}
q(\cdot) = \Delta(p(\cdot)). \label{eq:perturbation}
\end{align}
Denote the set of perturbations to be considered as $\mathcal{S}$.

\vspace{-0.1in}
\begin{figure}[h]
  \centering
  \includegraphics[width=0.43\textwidth]{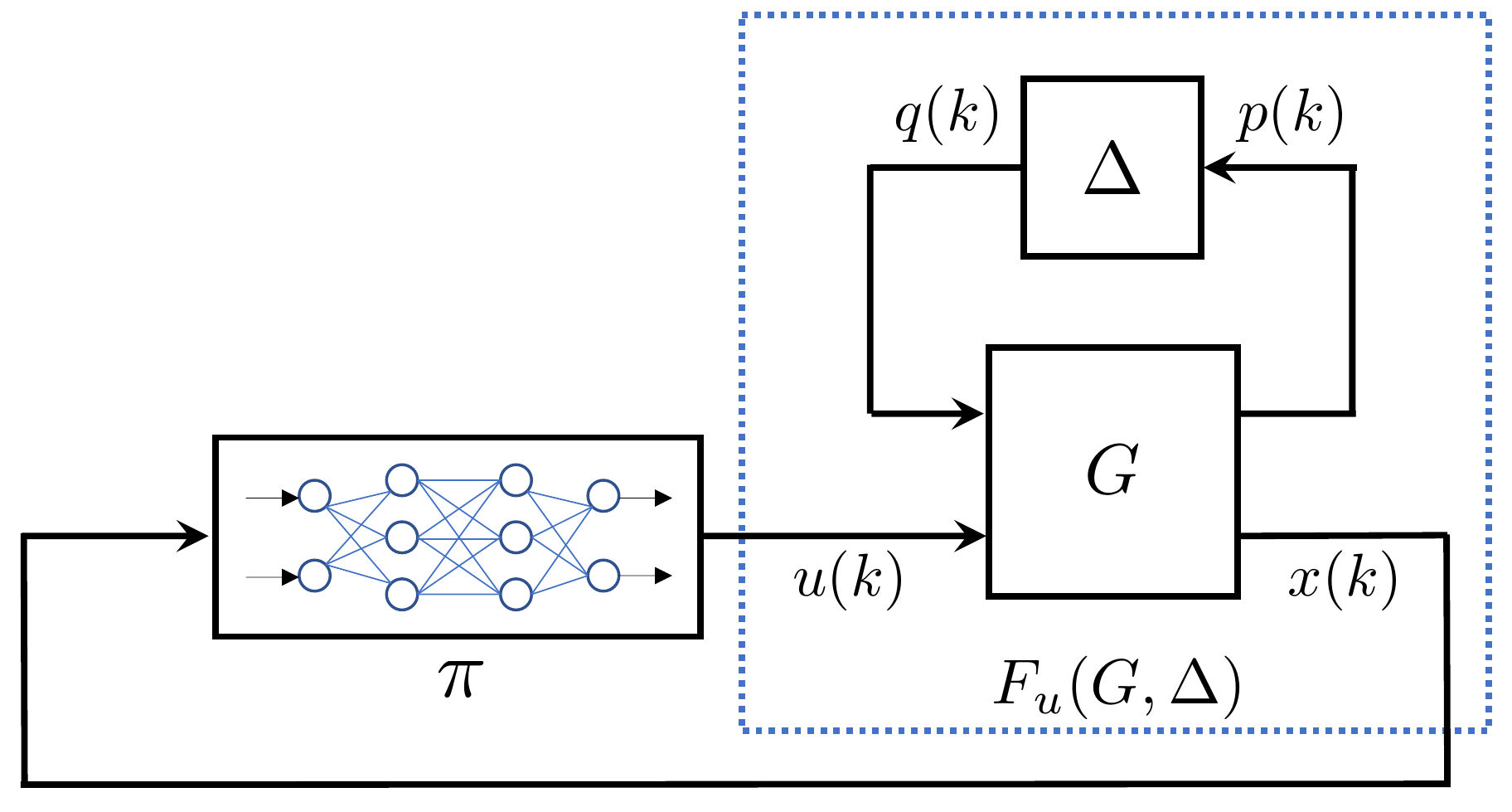}
  \caption{Feedback system with uncertain plant $F_u(G,\Delta)$ and  
    NN controller $\pi$}
  \label{fig:UncFeedback}    
\end{figure} 

\vspace{-0.15in}

\begin{assumption}\label{assume:zero_eq}
In this section, we assume (i) the equilibrium point $(x_*,u_*,v_*,w_*,p_*,q_*)$ of the feedback system is at the origin, and (ii) $0 = \Delta(0)$ for all $\Delta \in \mathcal{S}$.
%
Note that $\Delta$ is modeled as an
operator mapping inputs to outputs.  If $\Delta$ has an internal state
then there is an implicit assumption that it has zero initial condition.
\end{assumption}

Let $\chi(k;x_0,\Delta)$ denote the solution to the feedback system of $F_u(G,\Delta)$ and $\pi$ with $\Delta \in \mathcal{S}$ at
time $k$ from the initial condition $x(0)=x_0$\footnote{An
  input/output model is used for the perturbation $\Delta$ so that its
  internal state and initial condition is not explicitly considered.}.
Define the robust ROA associated with $x_*$ as follows.
\begin{definition}
  The robust ROA of the feedback system with the uncertain
  plant $F_u(G,\Delta)$ and NN $\pi$ is defined as:
  \vspace{-0.1in}
  \begin{align}
    \mathcal{R} := \{x_0 \in \R^{n_G}: \lim_{k \rightarrow \infty} \chi(k; x_0, \Delta) = x_* \ \forall \Delta \in \mathcal{S}\}. 
  \end{align}
\end{definition}

\vspace{-0.15in}
The objective is to prove the uncertain feedback system is asymptotically
stable and, if so, to find the largest estimate of the robust ROA
using an ellipsoidal inner approximation.

\subsection{Integral Quadratic Constraints}
The perturbation can represent various types of  uncertainty
\cite{Megretski:97}, \cite{Veenman:16}, including saturation, time
delay, unmodeled dynamics, and slope-restricted nonlinearities. The
input-output relationship of $\Delta$ is characterized with an
integral quadratic constraint (IQC), which consists of a `virtual'
filter $\Psi_\Delta$ applied to the input $p$ and output $q$ of
$\Delta$ and a constraint  on the output $r$ of
$\Psi_\Delta$. The filter $\Psi_\Delta$ is an LTI system of the form:
\begin{subequations}\label{eq:filterdyn}
\begin{align}
\psi(k+1) &= A_\Psi \ \psi(k) + B_{\Psi 1} \ p(k) + B_{\Psi 2} \ q(k) \\
r(k) &= C_\Psi \ \psi(k) + D_{\Psi 1} \ p(k) + D_{\Psi 2} \ q(k)
\label{eq:r_def}\\
\psi(0) &= 0
\end{align}
\end{subequations}
where $\psi(k) \in \R^{n_\psi}$ is the state, $r(k) \in \R^{n_r}$ is the
output, and $A_\Psi$ is a Schur matrix. The state matrices have compatible dimensions. The dynamics of $\Psi_\Delta$ can be compactly denoted by $\left[ \begin{array}{c|cc} A_\Psi & B_{\Psi 1}  & B_{\Psi 2} \\ \hline C_\Psi & D_{\Psi 1}  & D_{\Psi 2} \end{array} \right]$. By $(p_*, q_*)=0$ from Assumption~\ref{assume:zero_eq}, the equilibrium state
$\psi_* \in \R^{n_\psi}$ of \eqref{eq:filterdyn} is also zero.

The
Lyapunov analysis in the next subsection makes use of
time-domain IQCs as defined next:
\begin{definition}
  Let $\Psi_\Delta \in \mathbb{RH}_{\infty}^{n_r \times (n_p + n_q)}$
  and $M_\Delta \in \mathbb{S}^{n_r}$ be given. A bounded, causal
  operator $\Delta :\ell^{n_p}_{2e} \rightarrow \ell^{n_q}_{2e}$
  satisfies the time domain IQC defined by
  $(\Psi_\Delta, M_\Delta)$ if the following inequality holds for all
  $p \in \ell_{2e}^{n_{p}}$, $q = \Delta(p)$ and for all $N \ge 0$
  \begin{align} \label{eq:hardIQC}
    \sum_{k=0}^N r(k)^\top M_\Delta r(k) \ge 0.
  \end{align}
\end{definition}

The notation $\Delta \in$ IQC$(\Psi_\Delta, M_\Delta)$ indicates that
$\Delta$ satisfies the IQC defined by $\Psi_\Delta$ and
$M_\Delta$. Therefore, the precise relation (\ref{eq:perturbation}), for
analysis, is replaced by the constraint \eqref{eq:hardIQC} on $r$.
The QC proposed in Lemma~\ref{lemma:sectorQC} is a
special instance of a time-domain IQCs. Specifically,
Lemma~\ref{lemma:sectorQC} defines a QC that holds
at each time step $k$ and hence the inequality also holds summing over
any finite horizons. This is referred to as the offset local sector IQC.

The time-domain IQCs, as defined here, hold on any finite horizon
$N\ge 0$.  These are typically called ``hard IQCs''
\cite{Megretski:97}.  IQCs can also be defined in the frequency domain
and equivalently expressed as time-domain constraints over an infinite
horizon ($N=\infty$).  These are called soft IQCs. Although this paper focuses
on the use of hard IQCs, it is possible to also incorporate soft
IQCs \cite{IANNELLI2019,Fetzer2018_softIQC, YinBackward,YinForward}.

\vspace{-0.2in}
\subsection{Lyapunov Condition}
Let $\zeta := \bsmtx x \\ \psi \esmtx \in \R^{n_\zeta}$ define the
extended state vector, $n_\zeta = n_G + n_\psi$, whose dynamics are
\begin{subequations} \label{eq:extendsys}
    \begingroup
	\allowdisplaybreaks
\begin{align}
\zeta(k+1) &= \mathcal{A} \ \zeta(k) + \mathcal{B} \bmat{q(k) \\ u(k)} \\
r(k) &= \mathcal{C} \ \zeta(k) + \mathcal{D} \bmat{q(k) \\ u(k)} \\
u(k) &= \pi(x(k))
\end{align}
\endgroup
\end{subequations}
where the state-space matrices are
    \begingroup
	\allowdisplaybreaks
\begin{align}
\mathcal{A} = \bmat{A_G & 0 \\ B_{\Psi 1} C_G & A_\Psi}, \ \mathcal{B} =\bmat{B_{G1} & B_{G2} \\B_{\Psi 1}D_{G1} + B_{\Psi2} & B_{\Psi 1} D_{G2} }, \nonumber \\
\mathcal{C} = \bmat{D_{\Psi 1}C_G & C_\Psi}, \ \mathcal{D} = \bmat{D_{\Psi 1} D_{G1} + D_{\Psi 2} & D_{\Psi 1} D_{G 2}}. \nonumber
\end{align}
\endgroup
Let $\zeta_* := \bsmtx x_* \\ \psi_* \esmtx = 0$ define the equilibrium
point of the extended system \eqref{eq:extendsys}. Since IQCs
implicitly constrain the input $p$ of the extended system
\eqref{eq:extendsys}, the response of the extended system subject to
IQCs ``covers'' the behaviors of the original uncertain feedback
system. The following theorem provides a method for
inner-approximating the robust ROA by
performing analysis on the extended system subject to IQCs.

\begin{theorem}
  \label{thm:RobustLyap}
  Consider the feedback system of an uncertain plant $F_u(G,\Delta)$ in
  \eqref{eq:UncertainSys}--\eqref{eq:perturbation}, and the NN $\pi$
  in \eqref{eq:NNlong} with zero equilibrium point
  $(\zeta_*, u_*, v_*, w_*, p_*, q_*)$. Assume $\Delta \in$
  IQC$(\Psi_\Delta, M_\Delta)$ with $\Psi_\Delta$ and $M_\Delta$
  given. Let
  $\bar{v}^1 \in \R^{n_1}$, $\underline{v}^1:=-\bar{v}^1$, and
  $\alpha_\phi, \beta_\phi \in \R^{n_\phi}$ be given vectors
  satisfying Property~\ref{prop:NNsector} for the NN, and define matrices
    \begingroup
	\allowdisplaybreaks
\begin{align*}
    R_V &= \left[ \begin{array}{ccc} I_{n_\zeta} &0&0\\ \hline 0 & 0& I_{n_{q}}\\N_{u\zeta} & N_{uw}&0 \end{array} \right], \ N_{u\zeta} = [N_{ux}, 0_{n_u \times n_\psi}],\\
    R_{\phi} &= \bmat{N_{v\zeta} & N_{vw} & 0 \\0 & I_{n_\phi} & 0}, \ N_{v\zeta} = [N_{vx}, 0_{n_{\phi} \times n_\psi}],\\
    \mathcal{W}_i^1 &= \bmat{W^1_i & 0_{1\times n_\psi}}, \ W^1_i \ \text{is the $i^{th}$ row of} \ W^1.
  \end{align*}
  \endgroup
  If there exists a matrix $P \in \mathbb{S}_{++}^{n_\zeta}$, and vector
  $\lambda \in \R^{n_\phi}$ with $\lambda \ge 0$ such that
  \begin{subequations}
    \begin{align}
      &R_V^\top \bmat{\mathcal{A}^\top P \mathcal{A} - P & \mathcal{A}^\top P \mathcal{B} \\ \mathcal{B}^\top P \mathcal{A} & \mathcal{B}^\top P \mathcal{B}} R_V  + R_{\phi}^\top \Psi_\phi^\top M_\phi(\lambda) \Psi_\phi R_{\phi} \nonumber \\
      & ~~~~~~ + R_V^\top \bmat{\mathcal{C} & \mathcal{D}}^\top M_\Delta \bmat{\mathcal{C} & \mathcal{D}} R_V < 0, \label{eq:dissineq} \\
      &\bmat{(\bar{v}^1_i)^2 & \mathcal{W}^1_i \\ \mathcal{W}^{1\top}_i & P}  \ge 0, \ i = 1,\cdots, n_1, \label{eq:robsetcontain}
    \end{align}
  \end{subequations}
  then: (i) the feedback system comprising 
  $F_u(G,\Delta)$ and $\pi$ is locally stable around $x_*$ for any
  $\Delta\in$~IQC$(\Psi_\Delta,M_\Delta)$, and (ii) the intersection of
  $\mathcal{E}(P,\zeta_*)$ with the hyperplane $\psi = 0$, i.e.
  $\mathcal{E}(P_{x},x_*)$ where $P_x\in\R^{n_G \times n_G}$ is the upper left block of $P$,
 is an
  inner-approximation to the robust ROA.
\end{theorem}
\begin{proof}
  As in the proof of Theorem~\ref{thm:NominalLyap}, feasibility of
  \eqref{eq:robsetcontain} implies that if
  $\zeta(k) \in \mathcal{E}(P,\zeta_*)$ then
  $v^1(k) \in [\underline{v}^1, \bar{v}^1]$ and hence the offset local
  sectors conditions are valid. Since the LMI in \eqref{eq:dissineq}
  is strict, there exists $\epsilon>0$ such that left/right
  multiplication of the LMI by
  $[(\zeta(k)-\zeta_*)^\top, (w_\phi(k) - w_{*})^\top, (q(k) -
  q_*)^\top]$ and its transpose yields:
{\small
  \begin{align}
    &\Bigg[ \star \Bigg]^\top \bmat{\mathcal{A}^\top P \mathcal{A} - P & \mathcal{A}^\top P \mathcal{B} \\ \mathcal{B}^\top P \mathcal{A} & \mathcal{B}^\top P \mathcal{B}} \left[ \begin{array}{c} \zeta(k) - \zeta_*\\\hline q(k)-q_* \\ u(k)-u_* \end{array} \right]  \nonumber \\
    &+ \Bigg[ \star \Bigg]^\top \bmat{\mathcal{C} & \mathcal{D}}^\top M_\Delta \bmat{\mathcal{C} & \mathcal{D}}\left[ \begin{array}{c} \zeta(k) - \zeta_*\\\hline q(k)-q_* \\ u(k)-u_* \end{array} \right]   \nonumber \\
    &+ \Big[ \star \Big]^\top \Psi_\phi^\top M_\phi(\lambda) \Psi_\phi \bmat{v_\phi(k)-v_*\\w_\phi(k) - w_*}  
    \le -\epsilon \| \zeta(k)-\zeta_* \|^2, \nonumber 
  \end{align}}
  \noindent Define the Lyapunov function
  $V(\zeta): = (\zeta - \zeta_*)^\top P (\zeta - \zeta_*)$, and use
   \eqref{eq:extendsys}
  to show:
  \begin{align*}
    &V(\zeta(k+1)) - V(\zeta(k)) + r(k)^\top M_\Delta r(k)  
      \\
    &+ \Big[ \star \Big]^\top  \Psi_\phi^\top M_\phi(\lambda) \Psi_\phi 
    \bmat{v_\phi(k)-v_{*}\\w_\phi(k)-w_{*}}
    \le -\epsilon \| \zeta(k)-\zeta_* \|^2 .
  \end{align*}
  Sum this inequality from $k=0$ to any finite time $N\ge 0$.  The
  third and fourth term on the left side will be $\ge 0$ by the local
  sector conditions and the IQC. This yields:
  \begin{align*}
     V(\zeta(N+1)) - V(\zeta(0))
     \le - \sum_{k=0}^N \epsilon \| \zeta(k)-\zeta_* \|^2 .
  \end{align*}
  Thus if $\zeta(0)\in \mathcal{E}(P,\zeta_*)$ then
  $\zeta(k)\in \mathcal{E}(P,\zeta_*)$ for all $k\ge 0$.  Moreover,
  this inequality implies that $\zeta(N)\rightarrow \zeta_*$ as
  $N\rightarrow \infty$. The initial condition for the virtual filter
  is $\psi(0)=0$ so that $\zeta(0)\in \mathcal{E}(P,\zeta_*)$ is
  equivalent to $x(0) \in \mathcal{E}(P_x,x_*)$.  Hence
  $\mathcal{E}(P_{x},x_*)$ is an inner approximation for the ROA. 
\end{proof}

For a particular perturbation $\Delta$ there is typically a class of
valid time-domain IQCs defined by a fixed filter $\Psi_\Delta$ and a
matrix $M_{\Delta}$ drawn from a constraint set $\mathcal{M}_\Delta$.
Therefore when formulating an optimization problem, along with $P$ and
$\lambda$, we can treat $M_\Delta \in \mathcal{M}_\Delta$ as an
additional decision variable to reduce conservatism. In this paper, the set
$\mathcal{M}_\Delta$ is restricted to one that is described by LMIs
\cite{Veenman:16}.  Using trace$(P_x)$ as the cost function to minimize along with the LMIs developed before, we have the following optimization to compute the ``largest'' ROA inner-approximation:
\begin{align}\label{eq:final_opt}
    &\min_{P, \lambda, M_\Delta} \ \text{trace$(P_x)$} \ \text{s.t.} \ \eqref{eq:dissineq} \ \text{and} \ \eqref{eq:robsetcontain} \ \text{hold,} 
\end{align}
which is convex in $(P, \lambda, M_\Delta)$. The strict inequality in \eqref{eq:dissineq} can be enforced by either replacing $<0$ with $\leq -\epsilon I$ with $\epsilon = 1 \times 10^{-6}$, or solving \eqref{eq:final_opt} with a non-strict inequality $\leq 0$, and checking if the constraint is active afterwards.

\vspace{-0.2in}

\subsection{IQCs for Combined Activation Functions $\phi$}
Now that we have the general framework that merges Lyapunov theory with IQCs, we will revisit the problem of describing the activation functions $\phi$ using more general tools. Recall that offset local sector QCs have been used in Sections~\ref{sec:nomanaly} and \ref{sec:robustanaly} to bound activation functions $\phi$. However, these local sectors fail to incorporate slope bounds of $\phi$. In this subsection, in addition to local sectors, we will use off-by-one IQCs~\cite{Lessard2015} to capture the slope information of $\phi$ to achieve less conservative ROA inner-approximations.

Besides the local sector bound $\alpha_\phi, \beta_\phi$, the bounds $\underline{v}, \overline{v}$ on activation input $v_\phi$ can also be used to compute the local slope bounds $[m_\phi, L_\phi]$ of $\phi$, with $m_\phi, L_\phi \in \R^{n_\phi}$. For example, $\phi_i(v_{\phi,i}) := \tanh(v_{\phi,i})$ restricted to the interval $[\underline{v}_i, \overline{v}_i]$ for $i=1,...,n_\phi$ satisfies the local slope bound $[m_{\phi,i}, L_{\phi,i}]$  with $m_{\phi,i} := \min\left(\frac{d \tanh (v_i)}{d v_i}|_{v_i = \underline{v}_i}, \frac{d \tanh (v_i)}{d v_i}|_{v_i = \overline{v}_i}\right)$, and $L_{\phi,i} := 1$. If $w_\phi = \phi(v_\phi)$ with $v_\phi(k) \in [\underline{v}, \overline{v}]$, then $\phi$ also satisfies the hard IQC defined by $(\Psi_\text{off}, M_\text{off})$, where 
    \begingroup
	\allowdisplaybreaks
\begin{align}
    &\Psi_\text{off}  :=  \left[ \begin{array}{c|cc} 0_{n_\phi} & -\text{diag}(L_\phi)  &I_{n_\phi} \\ \hline I_{n_\phi} & \text{diag}(L_{\phi})  &-I_{n_\phi} \\ 0_{n_\phi} & -\text{diag}(m_{\phi})  &I_{n_\phi} \end{array} \right],  \nonumber \\
     &M_\text{off}(\eta) := \bmat{0_{n_\phi} & \text{diag}(\eta) \\ 
             \text{diag}(\eta) & 0_{n_\phi}}, \ \text{for all} \ \eta \in \R^{n_\phi} \ \text{with} \ \eta \ge 0. \nonumber
\end{align}
\endgroup
This is the so-called ``off-by-one'' IQC \cite{Lessard2015}, which is a special instance of the Zames-Falb IQC~\cite{Zames1968,Kulkarni2002}. It provides constraints
that relate the activation at different time instances, e.g. between
$\phi_i(v_{\phi,i}(k))$ and $\phi_i(v_{\phi,i}(k+1))$ for any
$i=1,\ldots,n_\phi$. 

The analysis on the feedback system of $F_u(G,\Delta)$ and $\pi$ can be instead performed on the extended system made up by $G$, $\Psi_\Delta$ and $\Psi_\text{off}$ with additional constraints that $\Delta \in$ IQC$(\Psi_\Delta, M_\Delta)$, and $\phi$ satsifies the offset local sector and $\phi \in$ IQC$(\Psi_\text{off}, M_\text{off})$. However, since $\Psi_\text{off}$ introduces a number of $n_\phi$ states to the extended system,  the size of the corresponding Lyapunov matrix $P$ will increase from $\mathbb{S}_{++}^{n_x + n_\psi}$ to $\mathbb{S}_{++}^{n_x + n_\psi + n_\phi}$, which leads to longer computation time. The effectiveness of the off-by-one IQC is demonstrated in Section~\ref{sec:vehicle_example}.

\section{Examples}\label{sec:example}
In the following examples, the optimization~\eqref{eq:final_opt} is solved using MOSEK with CVX. \footnote{The code is available at \href{https://github.com/heyinUCB/Stability-Analysis-using-Quadratic-Constraints-for-Systems-with-Neural-Network-Controllers}{https://github.com/heyinUCB/Stability-Analysis-using-Quadratic-Constraints-for-Systems-with-Neural-Network-Controllers}.}

\vspace{-0.2in}
\subsection{Inverted pendulum}
Consider the nonlinear inverted pendulum example with mass $m = 0.15$ kg, length $l = 0.5$ m, and
friction coefficient $\mu = 0.5$ Nms$/$rad. The dynamics are:
\begin{align}
  \ddot{\theta}(t) &= \frac{mgl\sin(\theta(t)) - \mu \dot{\theta}(t) + \text{sat}(u(t))}{m l^2}, \label{eq:thetadot}
\end{align}
where $\theta$ is the angular position (rad) and $u$ is the control
input (Nm). The plant state is $x = [\theta, \dot{\theta}]$. The saturation function is defined as $\text{sat}(u) = \text{sgn}(u)\min(|u|,u_{\text{max}})$,
with $u_{\text{max}} = 0.7$~Nm. The
controller $\pi$ is obtained through a reinforcement learning process
using policy gradient \cite{Peters2008, Schulman2015,
  Kakade01}. During training, the agent decision making process is
characterized by a probability:
$u(k) \sim Pr(u(k) = u \ \vert \ x(k)=x)$ for all $u \in \R$ and
$x \in \R^2$ where the probability is a Gaussian distribution with
mean $\pi(x(k))$, and standard deviation $\sigma$. After training, the
policy mean $\pi$ is used as the deterministic controller
$u(k) = \pi(x(k))$. The controller $\pi$ is parameterized by a
2-layer, feedforward NN with $n_1 = n_2 = 32$ and $\tanh$ as the
activation function for both layers. The biases in the NN are set to
zero during training to ensure that the equilibrium point is $x_*=0$
and $u_*=0$. The dynamics used for training are the discretized
version of \eqref{eq:thetadot} with sampling time $dt = 0.02$~s.

We rearrange \eqref{eq:thetadot} into the 
form:
\begin{subequations}
    \begingroup
	\allowdisplaybreaks
\begin{align}
  \ddot{\theta}(t) &= \frac{-mgl q(t) + mgl \theta(t) - \mu \dot{\theta}(t) + \text{sat}(u(t))}{m l^2}, \\
  q(t) &= \Delta(\theta(t)) := \theta(t) - \sin(\theta(t)).
\end{align}
\endgroup
\end{subequations}
The static nonlinearity $\Delta(\theta)=\theta - \sin(\theta)$  is slope-restricted, and sector bounded.
If we assume that $\theta(k) \in [\underline{\theta}, \overline{\theta}]$ with 
$\bar{\theta} = -\underline{\theta} = 0.73$, then the nonlinearity is
slope-restricted in $[0, 0.2548]$, and sector bounded in $[0,
0.087]$. 
We also assume that
$v^1 \in [\underline{v}^1, \bar{v}^1]$ with
$\bar{v}^1 = -\underline{v}^1 = \delta_v \times 1_{32 \times 1}$ using $\delta_v = 0.1$. Both assumptions are
verified using the ROA inner-approximation. Two types
of IQCs are used to characterize the nonlinearity $\Delta(\cdot)$: an off-by-one IQC to capture the slope information, and a local
sector IQC to express the local sector bound. Only the local sector IQC is used to characterize the activation functions $\phi$. The saturation function is static and can also be described using a local sector
bound. Let $\bar{u}$ be the largest possible control command from
$\pi$ induced from the assumption that
$v^1 \in [\underline{v}^1, \bar{v}^1]$. Then the
saturation function  satisfies the local sector
$[\alpha, \beta]$, where $\alpha:=\frac{u_\text{max}}{\bar{u}}$ and
$\beta:=1$. 

Figure~\ref{fig:penROA} shows the
boundaries for the sets
$\{x: \underline{v}^1 \leq v^1 \leq \bar{v}^1\}$ and
$\{x: \underline{\theta} \leq \theta \leq \bar{\theta}\}$ 
with orange and brown lines, the ROA inner-approximation  with
a blue ellipsoid, and the phase portrait of the closed-loop system, with green and red curves representing trajectories inside and outside the ROA.
\vspace{-0.1in}
\begin{figure}[h]
	\centering
	\includegraphics[width=0.35\textwidth]{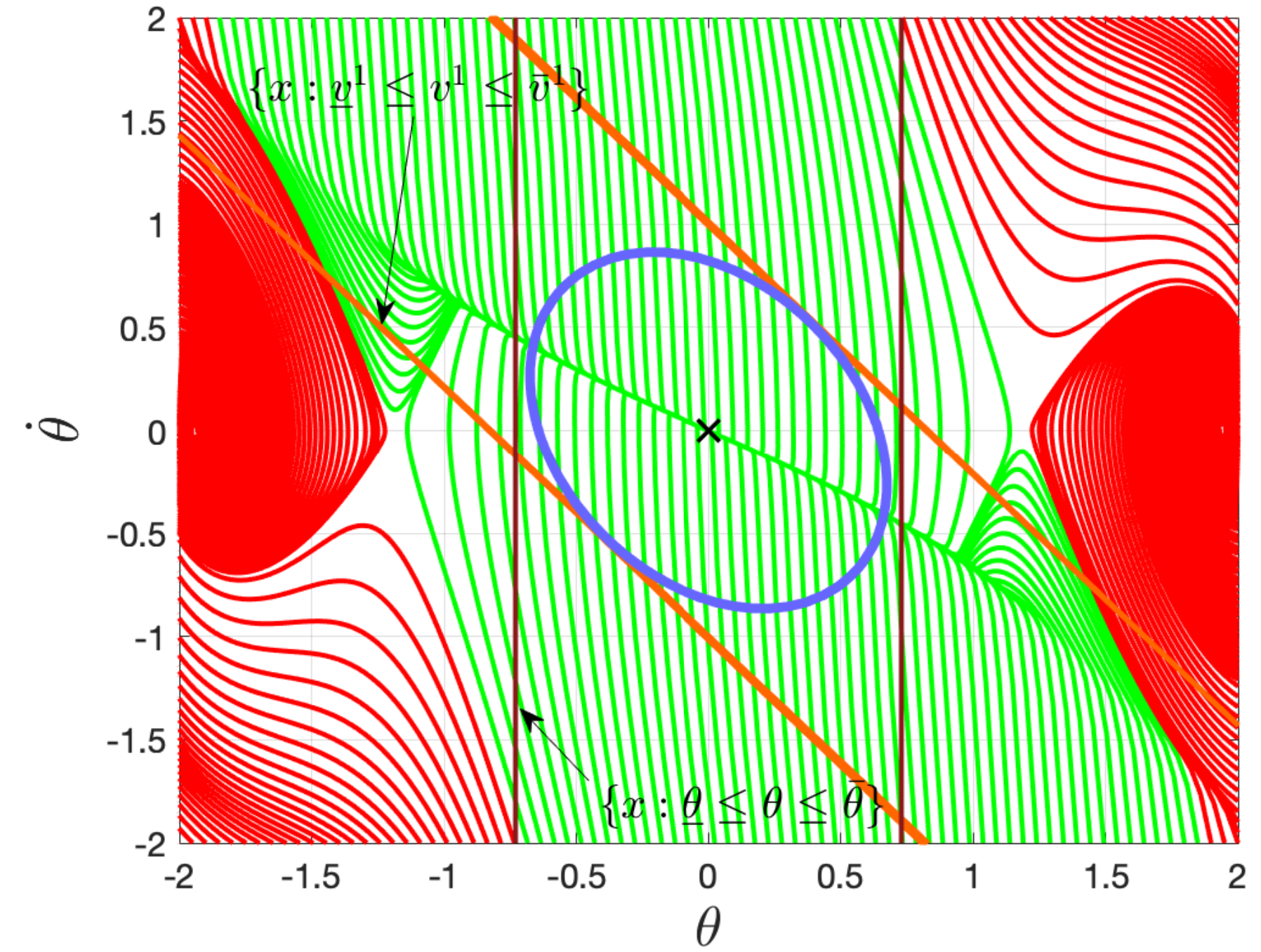}
	\caption{A ROA inner-approximation of the inverted pendulum}
	\label{fig:penROA}    
\end{figure}

\vspace{-0.25in}
\subsection{Vehicle lateral control}\label{sec:vehicle_example}
Consider the vehicle lateral dynamics from \cite{ALLEYNE:1997}:
\begingroup
\allowdisplaybreaks
{\small
\begin{align}\label{eq:vehidyn}
\bmat{\dot{e} \\ \ddot{e} \\ \dot{e}_\theta \\ \ddot{e}_\theta} = &\bmat{0 & 1 & 0&0 \\ 0 & \frac{C_{\alpha f} + C_{\alpha r}}{m U} & -\frac{C_{\alpha f} + C_{\alpha r}}{m} & \frac{a C_{\alpha f} - b C_{\alpha r}}{m U} \\ 0 & 0 & 0&1 \\ 0 & \frac{a C_{\alpha f} - b C_{\alpha r}}{I_z U} & -\frac{a C_{\alpha f} - b C_{\alpha r}}{I_z} & \frac{a^2 C_{\alpha f} + b^2 C_{\alpha r}}{I_z U}} \bmat{e \\ \dot{e} \\ e_\theta \\ \dot{e}_\theta}\nonumber \\
+&\bmat{0 \\ -\frac{C_{\alpha f}}{m} \\ 0 \\ -\frac{a C_{\alpha f}}{I_z}} u + \bmat{0\\ \frac{a C_{\alpha f} - b C_{\alpha r}}{m} - U^2 \\ 0 \\ \frac{a^2 C_{\alpha f} + b^2 C_{\alpha r}}{I_z}} c
\end{align}}
\endgroup
where $e$ is the perpendicular distance to the lane edge (m), and $e_\theta$ is the angle between the tangent to the straight section of the road and the projection of the vehicle's longitudinal axis (rad). Let $x = [e,\dot{e}, e_\theta, \dot{e}_\theta]^\top$ denote the plant state. The control $u$ is the steering angle of the front wheel (rad), the disturbance $c$ is the road curvature ($1/$m), and the parameters are: longitudinal velocity $U = 28$ m$/$s, front cornering stiffness $C_{\alpha f} = -1.232\times 10^5$ N$/$rad, rear cornering stiffness $C_{\alpha r} = -1.042\times 10^5$ N$/$rad, mass $m = 1.67 \times 10^3$ kg, moment of inertia $I_z = 2.1 \times 10^3$ kg$/\text{m}^2$, distances from vehicle center of gravity to front axle $a=0.99$ m and rear axle $b = 1.7$ m.

Again, the controller $\pi$ is obtained using policy gradient, and is parameterized by a 2-layer, feedforward NN, with $n_1 = n_2 = 32$ and $\tanh$ as the activation function for both layers. The training process uses a discretized version of \eqref{eq:vehidyn} with sampling time $dt = 0.02$~s and draws the curvature $c(k)$ at each time step from an interval $[-1/200, 1/200]$. The control command derived from $u(k) = \pi(x(k))$ enters the vehicle dynamics through a saturation function sat$(\cdot)$ with $u_\text{max} = \pi/6$.
Let $u_{\text{sat}} := \text{sat}(\pi(x))$ define the saturated control signal.

The analysis is performed for  a constant 
curvature $c \equiv 0$, resulting in a zero equilibrium state
$x_* =0$.  In the analysis problem,
on top of saturation, we also add a norm-bounded LTI uncertainty
$\Delta_{\text{LTI}} \in \mathbb{RH}_\infty$ with
$\norm{\Delta_\text{LTI}}_\infty~\leq~0.1$ to the control input. This
is used to assess the robustness of the NN controller against 
actuator uncertainty. As shown in Figure~\ref{fig:veh}, 
the actual input to the vehicle dynamics is 
\begin{align*}
    u_{\text{pert}}(k) = u_{\text{sat}}(k) + q(k), \ \text{and} \ 
    q(\cdot) = \Delta_{\text{LTI}}(u_{\text{sat}}(\cdot)).
\end{align*}
\vspace{-0.3in}
\begin{figure}[h]
	\centering
	\includegraphics[width=0.37\textwidth]{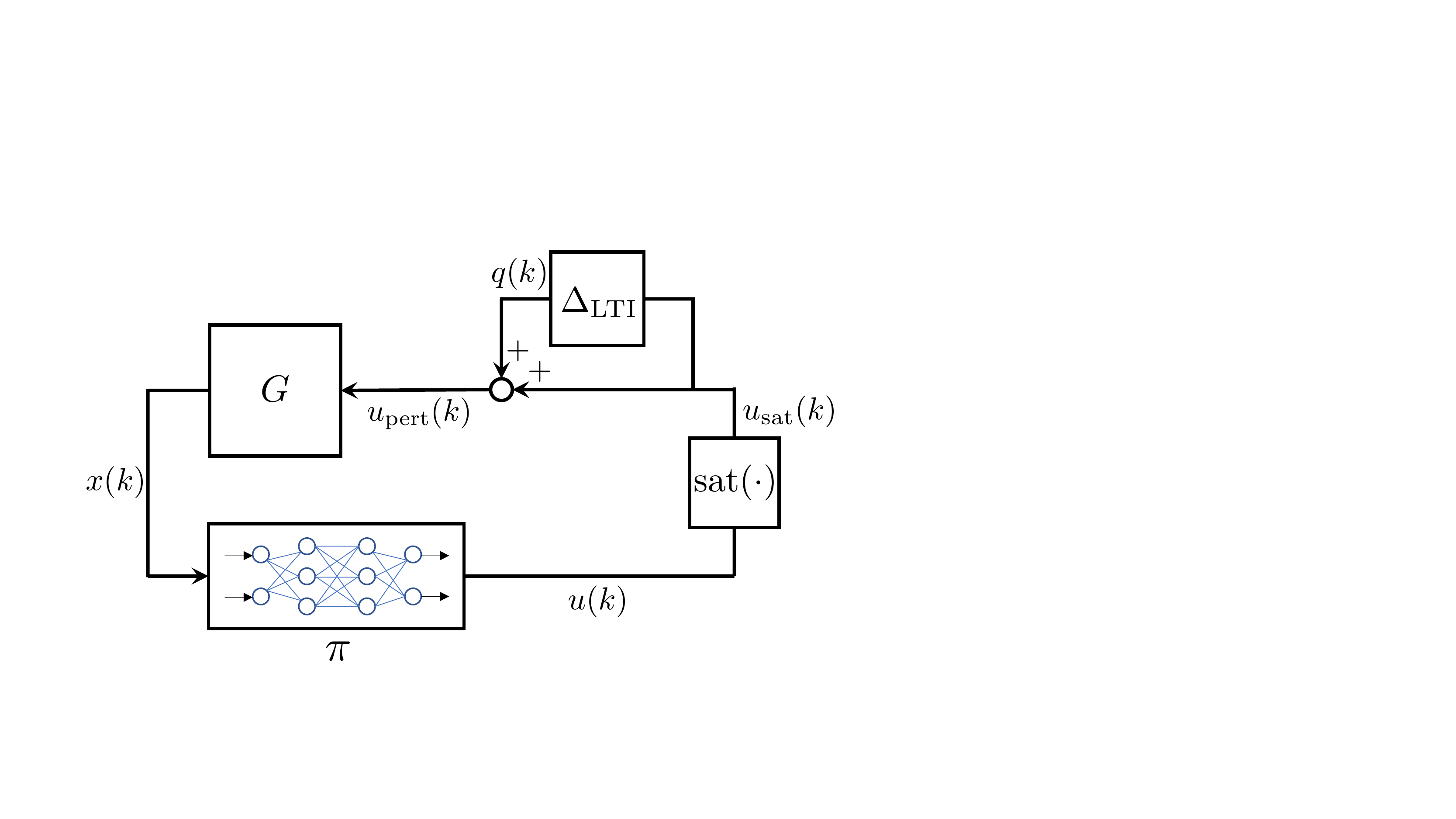}
	\caption{Uncertain vehicle system with actuator uncertainty}
	\label{fig:veh}    
\end{figure}
\vspace{-0.1in}

It is assumed that $v^1 \in [\underline{v}^1, \bar{v}^1]$, where
$\bar{v}^1  = - \underline{v}^1 = \delta_v \times 1_{32\times1}$ with $\delta_v = 0.6$. To show effectiveness of the off-by-one IQC, two experiments were carried out: one with only local sector IQC to describe $\phi$, and one with both local sector and off-by-one IQCs. The achieved trace$(P_x)$ for the two experiments are 4.4 and 2.9, respectively. Moreover, the achieved $\det(P_x^{-1})$ (proportional to the volume) for the experiments are $3.2\times 10^5$, and $1.1 \times 10^6$, respectively. Therefore, with the help of off-by-one IQC to sharpen the description of $\phi$, the second experiment achieves a larger ROA inner-approximation. It is also important to note that thanks to the off-by-one IQC, the SDP is able to tolerate looser local sector bounds. The largest value of $\delta_v$ such that the SDP is feasible is 0.67 for the first experiment, and 1.4 for the second experiment.

Figure~\ref{fig:vehROA} shows slices of
the ROA inner-approximation from the second experiment on the $e$--$\dot{e}$  and
$e_\theta$--$\dot{e}_\theta$ spaces. Specifically, these are
intersections of $\mathcal{E}(P_x, x_*)$ with the hyperplanes
$(e_\theta,\dot{e}_\theta) = (e_{\theta*}, \dot{e}_{\theta*})$ and
$(e,\dot{e}) = (e_*, \dot{e}_*)$, respectively, where
$x_* = [e_*, \dot{e}_*, e_{\theta*}, \dot{e}_{\theta*}]^\top$. The
slices are shown with blue ellipsoids. The boundary of the polytopic
set $\{x: \underline{v}^1 \leq v^1 \leq \bar{v}^1\}$ is shown with the
orange lines. The brown crosses represent the zero
equilibrium state $x_*$.
\vspace{-0.1in}
\begin{figure}[h]
  \centering
  \includegraphics[width=0.37\textwidth]{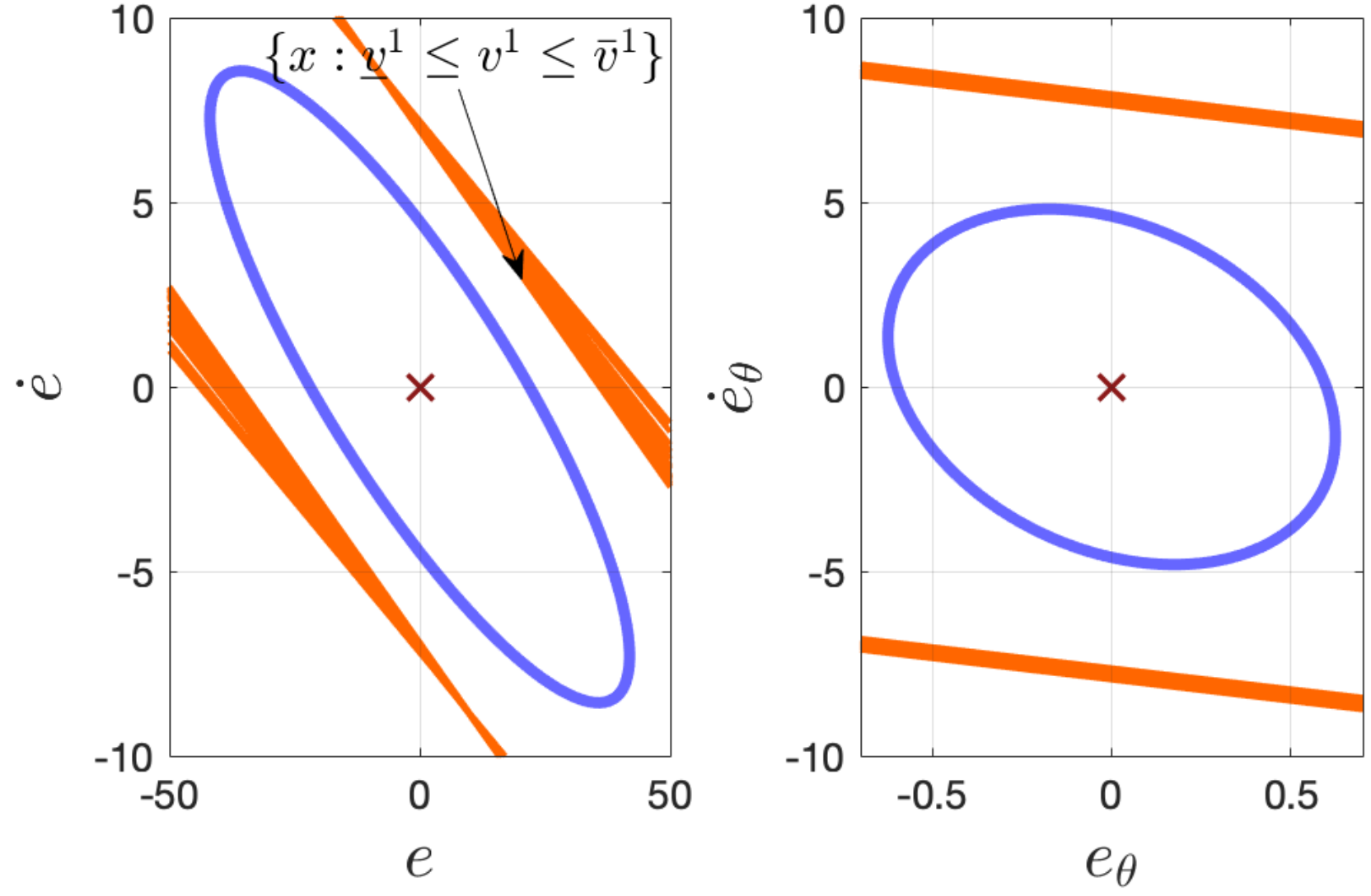}
  \vspace{-0.1in}
  \caption{ROA inner-approximation on the $e$--$\dot{e}$ and $e_\theta$--$\dot{e}_\theta$ spaces using both local sector and off-by-one IQCs with $\delta_v = 0.6$.}
  \label{fig:vehROA}    
\end{figure}



\vspace{-0.15in}
\bibliography{bibfile.bib} 
\bibliographystyle{IEEEtran}

\end{document}